\documentclass[10pt,a4paper]{article}
\usepackage{a4wide}
\usepackage[latin1]{inputenc}
\usepackage{amsmath}
\usepackage{amsfonts}
\usepackage{amssymb}
\usepackage{amsthm}
\usepackage{authblk}
\usepackage{thm-restate}

\newtheorem{lemma}{Lemma}

\newcommand{\BO}[1]{{\mathcal O}\left(#1\right)}

\newcommand{\BOM}[1]{\Omega\left(#1\right)}
\newcommand{\sort}[1]{\text{sort}\left({#1}\right)}

\newcommand{\E}[1]{\mathbb{E}\left[#1\right]}
\newcommand{\EC}[2]{\mathbb{E}\left[#1 \vert #2\right]}
\newcommand{\V}[1]{\text{Var}\left(#1\right)}

\renewcommand{\Pr}[1]{\text{Pr}\left(#1\right)}

\title{The Input/Output Complexity of Triangle Enumeration}
\date{}
\author[1]{Rasmus Pagh}
\author[2]{Francesco Silvestri\thanks{This work was done while Silvestri was visiting the IT University of Copenhagen.
}}
\affil[1]{IT University of Copenhagen, \texttt{pagh@itu.dk}}
\affil[2]{University of Padova, \texttt{silvest1@dei.unipd.it}}

\begin{document}
\maketitle

\begin{abstract}
We consider the well-known problem of enumerating all triangles of an undirected graph. Our focus is on determining the input/output (I/O) complexity of this problem. Let $E$ be the number of edges, $M<E$ the size of internal memory, and $B$ the block size. The best results obtained previously are sort$(E^{3/2})$ I/Os (Dementiev, PhD thesis 2006) and $\BO{E^2/(MB)}$ I/Os (Hu et al., SIGMOD 2013), where sort$(n)$ denotes the number of I/Os for sorting $n$ items. We improve the I/O complexity to $\BO{E^{3/2}/(\sqrt{M} B)}$ expected I/Os, which improves the previous bounds by a factor $\min(\sqrt{E/M},\sqrt{M})$. Our algorithm is cache-oblivious and also I/O optimal: We show that any algorithm enumerating $t$ distinct triangles must {\em always\/} use $\BOM{t/(\sqrt{M} B)}$ I/Os, and there are graphs for which $t=\BOM{E^{3/2}}$. Finally, we give a deterministic cache-aware algorithm using $\BO{E^{3/2}/(\sqrt{M} B)}$ I/Os assuming $M\geq E^\varepsilon$ for a constant $\varepsilon > 0$. Our results are based on a new color coding technique, which may be of independent interest. 
\end{abstract}

\section{Introduction}
Many kinds of information can be naturally represented as graphs, and algorithms for processing information in this format often need to consider small subgraphs such as triangles.
Examples of applications in which we need to enumerate all triangles in a graph are found in~\cite{WellesDC10} for studying social processes in networks,~\cite{BerryHLP11} for community detection,~\cite{FudosH97} for solving
systems of geometric constraints. See~\cite{HuTC13,Berry14} for further discussion and examples.

A classical example from database theory is the following.
A database is created to store information on salespeople and the products they sell.
Each product is characterized by a brand and a product type, e.g.~``ACME vacuum cleaner'', where each product type may be available in many brands.
An obvious representation in a relational database is a single table {\tt Sells(salesperson,brand,productType)}.
However, suppose that a salesperson is characterized by a set $B$ of brands and a set $T$ of product types, and she sells all available products in $B\times T$.
Then {\tt Sells} is not in 5th normal form\footnote{The 5th normal form reduces redundancy in relational databases recording multi-valued facts. Intuitively, a table is in 5th normal form if it cannot be reconstructed from smaller tables using equijoins~\cite{Kent83}.}, so to avoid anomalies it should be decomposed into three tables, one for each pair of attributes, whose natural join is equal to {\tt Sells}.
Viewing each table as a bipartite graph with vertices corresponding to attribute values, computing {\tt Sells} is exactly the task of enumerating all triangles in the union of these three graphs.
In other words, to be able to compute the join of three tables that are in 5th normal form we must solve the triangle enumeration problem.
Surprisingly, it seems that the challenge of doing this in an I/O-efficient way was not addressed in the database community until the SIGMOD 2013 paper of Hu, Tao and Chung~\cite{HuTC13}, though we note that a pipelined nested loop join does a good job when the edge set almost fits in memory.

In the context of I/O-efficient algorithms it is natural to not require the {\em listing\/} of all triangles to external memory.
Rather, we simply require that the algorithm {\em enumerates\/} all triangles.
More precisely, it suffices that for each triangle $\{v_1,v_2,v_3\}$ the algorithm makes exactly one call to a procedure {\tt emit}$(\cdot,\cdot,\cdot)$ with parameters $(v_1,v_2,v_3)$ at a point of time during the computation where all edges $\{v_1,v_2\}$, $\{v_2,v_3\}$, and $\{v_1,v_3\}$ are present in internal memory. 
Focusing on enumeration rather than listing is in line with the way the I/O complexity of algorithms in database systems is usually accounted for, where pipelining of operations may mean that it is not necessary to materialize an intermediate result.
The same is true for other applications in which enumerating all triangles is a preprocessing step.
Since each triangle is emitted at exactly one point in time there is no need for a separate duplicate elimination step.

The algorithm for triangle listing in \cite{HuTC13} can be easily adapted to solve the enumeration problem.
We recently learned that Hu et al.~also make this observation in the journal version~\cite{HuTC14} of their SIGMOD 2013 paper.
However, the algorithm requires $\BO{E^2/(MB)}$ I/Os for enumerating all triangles. 
We note that the I/O complexity corresponds to $E/M$ scans of the edge set. 
The main message of this paper is that it is possible to improve this I/O complexity by a factor $\sqrt{E/M}$, which is significant whenever the data size is much larger than internal memory. 
Our contributions are the following:
\begin{itemize}
\item We present a randomized triangle enumeration algorithm that is cache-oblivious~\cite{FrigoLPR12}, and improves the I/O complexity of previous algorithms by an expected factor $\min\left(\sqrt{E/M},\sqrt{M}\right)$. This is significant for large graphs in which $E\gg M \gg 1$.
\item We present a deterministic and cache-aware triangle enumeration algorithm with the same asymptotic I/O complexity under the mild assumption $M\geq \sqrt{E}$.
\item We show that the number of I/Os of our algorithms is within a constant factor of the best possible under the assumption that each triangle output must be ``witnessed'' by edges stored in internal memory. 
A similar result has been independently achieved in~\cite{HuTC14}.
\end{itemize}
Formal statements can be found in Section~\ref{sec:ourresults}.

\subsection{State of the art}

Algorithms for memory hierarchies, in particular in the external memory model~\cite{AggarwalV88}, have been widely investigated in the last years and we refer to the excellent survey by Vitter~\cite{Vitter08} for a complete overview of the state of the art. 
Cache-oblivious algorithms have been introduced by Frigo et al.~\cite{FrigoLPR99,FrigoLPR12} and are algorithms that do not use in their specifications the parameters describing the memory hierarchy, but still exhibit an optimal or quasi-optimal I/O complexity.

The triangle listing and enumeration problems are equivalent in flat memory (e.g., the RAM model) since the cost of writing in memory all the enumerated triangles is asymptotically no larger than the cost of triangle generation. 
However, this is not the case when external storage is used: 
the cost of writing triangles can significantly increase the I/O complexity in graphs with a large number of triangles. 
  
Several previous papers have considered the problem of listing triangles in the external memory model. 
Before considering these papers, we observe that since triangle enumeration can be expressed as a natural join of three relations, it is possible to use two block-nested loop joins (in a pipelined fashion) to solve the problem incurring $\BO{(E/M)^2 E/B} = \BO{E^3/(M^2 B)}$ I/Os.

The first two works dealing explicitly with triangle listing in external memory are due to Menegola~\cite{Menegola10} and Dementiev~\cite{Dementiev2007}, which give algorithms using, respectively, $\BO{E+E^{1.5}/B}$ and $\BO{(E^{1.5}/B) \log_{M/B} (E/B)}$ I/Os. 
Both algorithms incur a large number of I/Os and have weak temporal locality of reference since their bounds have at most a logarithmic dependency on the memory size $M$.
Using graph partitioning ideas, Chu and Cheng~\cite{ChuC12} improved the bound to $\BO{E^2/(MB)+t/B}$ for a class of graphs, where $t$ is the number of returned triangles. 
This bound improves the previous ones as soon as $M=\BOM{\sqrt{E}}$ and is the first to be output sensitive. 
The class of graphs handled are those for which each subgraph generated by the partitioning fits in memory.
Hu, Tao and Chung~\cite{HuTC13} provided an algorithm reaching the same bound, using very different techniques, 
working for {\em arbitrary\/} graphs.
This improves the algorithm based on block-nested loop joins by a factor $E/M$.
It is argued in~\cite{HuTC13} (and elaborated in the full version~\cite{HuTC14}) that their algorithm is near-optimal in the sense that it cannot be significantly improved for {\em all\/} combinations of $E$ and $M$.
However, the argument leaves open the question of whether a significant improvement can be obtained when $E \gg M$.
In contrast, we show matching upper and lower bounds for all combinations of $E$ and $M$. 

As mentioned above, a lower bound on the I/O complexity of triangle enumeration has independently been shown in the unpublished journal article~\cite{HuTC14}.
Although the main interest of the paper is in the listing problem, it provides a $\BOM{E^{3/2}/(\sqrt{M}B) + E/B}$ lower bound on the I/O complexity that applies also to the enumeration problem. 
However, in this paper we extend this result to be a best-case lower bound and to be output sensitive, using a shorter and arguably simpler argument. 
That is, we show that the I/O complexity of any algorithm for enumerating $t$ triangles is $\BOM{t/(\sqrt{M}B) + t^{2/3}/B}$. Both bounds apply to algorithms that, intuitively, manage edges and vertices as atomic information.

We recall that triangle listing has been widely studied in other models (there is no distinction between enumeration and listing in these works). 
The relations between listing and other problems have been widely investigated, see for instance  Williams and Williams~\cite{WilliamsW10} for a reduction to  matrix multiplication, and Jafargholi and Viola~\cite{JafargholiV13} for 3SUM/3XOR. 
Parallel algorithms for triangle listing have been addressed in the MapReduce framework by Afrati et al.~\cite{AfratiSSU13}, and by Suri and Vassilvitskii~\cite{SuriV11}. 
Triangle listing in certain classes of random graphs has been addressed recently by Berry et al.~\cite{Berry14} to explain the empirically good behavior of simple triangle listing algorithms. 
For the related problem of counting the number of triangles in a graph, we refer to~\cite{KolountzakisMPT12} and references therein.

\subsection{Our results}\label{sec:ourresults}
Our first main result is a cache-oblivious algorithm for triangle enumeration. 
In a cache-oblivious algorithm  no variables dependent on hardware parameters, such as internal memory size and block length, need to be tuned to achieve optimality (or quasi optimality).
The cache-oblivious algorithm  is inspired by a recursive approach proposed by Jafargholi and Viola~\cite{JafargholiV13}, in the context of output sensitive triangle listing in the RAM model. We prove the following claim.
\begin{restatable}{thm}{thcacheoblivious}
\label{th:cobl}
Assume $E\geq M$. Then there exists a cache-oblivious randomized algorithm for triangle enumeration using $\BO{\frac{E^{3/2}}{\sqrt{M}B}}$ I/Os in expectation  and $\BO{E}$ words on disk.
\end{restatable}
\noindent By a property of cache-oblivious algorithms~\cite{FrigoLPR12}, we have that the claimed I/O complexity applies to each  level of a multilevel cache with an LRU replacement policy.

Our second result is a deterministic cache-aware algorithm with the same I/O complexity as the cache-oblivious algorithm, under the assumption that internal memory has size at least $\sqrt{E}$.
This is a reasonable assumption in practice if we are concerned with a graph stored on hard disk or on solid-state drive and $M$ is the capacity of the RAM.
The algorithm is based on the derandomization of a simple cache-aware algorithm, described in Section~\ref{sec:aware}. 
The derandomization  uses an idea introduced in~\cite{JafargholiV13}, though we present a more refined greedy approach that preserves the exponent $3/2$ of the algorithm.  
We conjecture that with some technical adjustments the derandomization can be also applied to the cache-oblivious algorithm. 
\begin{restatable}{thm}{thdeterministic}\label{th:det}
Assume $E\geq M\geq {E^{\epsilon}}$, for an arbitrary constant $\epsilon>0$. 
Then there exists a deterministic, cache-aware algorithm for triangle enumeration that uses $\BO{\frac{E^{3/2}}{\sqrt{M}B}}$ I/Os and $\BO{E}$ words on disk in the worst case.
\end{restatable}

Finally, we prove that the I/O complexity of our algorithms is optimal in the external memory model. 
We assume that information on an edge requires at least one memory word:
this assumption is similar to the indivisibility assumption~\cite{ArgeM99} which is usually required for deriving lower bounds on the I/O complexity, or to the witnessing class of the aforementioned lower bound in~\cite{HuTC14}.
With respect to this bound, we remark that our lower bound applies also in the best-case and it is output sensitive. 
\begin{restatable}{thm}{thlowerbound}
\label{th:lb}
For any input graph, an algorithm that enumerates $t$ distinct triangles requires, even in the best case, $\BOM{\frac{t}{\sqrt{M}B}+\frac{t^{2/3}}{B}}$ I/Os .
\end{restatable}

The above lower bound on the I/O complexity applies also in the case of a weak definition of the triangle enumeration problem, which requires an algorithm to make \emph{at least} one call to the procedure {\tt emit}$(\cdot,\cdot,\cdot)$ for each triangle. Algorithms for the weak triangle enumeration problem may not be able to compute the exact number of triangles in a graph, while this is not the case of our algorithms.

Although the work of an algorithm is not the main complexity measure in the external memory model, we remark that all our algorithms are also work optimal: indeed, it can be easily  proved that each algorithm performs $\BO{E^{3/2}}$ operations in the worst case, matching the naive $\BOM{t}$ lower bound for enumerating $t$ triangles when $t=\BOM{E^{3/2}}$.

The paper is organized as follows. 
Section~\ref{sec:aware} gives a simple cache-aware randomized algorithm. 
Section~\ref{sec:cobl}  describes the claimed cache-oblivious randomized algorithm. 
The deterministic  algorithm is then proposed in Section~\ref{sec:derand} by derandomizing the previous cache-aware randomized algorithm. 
Section~\ref{sec:lb} gives the lower bound on the I/O complexity.
We conclude the paper with some final comments in Section~\ref{sec:concl}.

\subsection{Preliminaries}\label{sec:prelim}
We study our algorithms in the \emph{external memory model}~\cite{AggarwalV88}, which consists of an internal memory of $M$ words and of an external memory of unbounded size. 
The processor can only use data stored in internal memory and move data from the two memories in chunk of consecutive $B$ words. 
The \emph{I/O complexity} of an algorithm is defined as the number of input/output blocks performed by the algorithm. 
We denote the I/O complexity of sorting $n$ entries with $\text{sort}(n)=\BO{\frac{n \log (n/B)}{B \log M}+\frac{n}{B}}$~\cite{Vitter08}.

A \emph{cache-oblivious} algorithm is an algorithm that does not use in its specification the parameters describing the memory hierarchy (i.e., $M$ and $B$ in our model), but still exhibits an optimal or quasi-optimal I/O complexity. 
An algorithm that does use at least one of these parameters is said \emph{cache-aware}.
In the context of cache-oblivious algorithms, we assume that block transfers between internal and external memories are automatically managed by an optimal replacement policy. 
However, it can be shown~\cite[Lemma 6.4]{FrigoLPR12}  that optimality with an optimal replacement policy implies an optimal number of I/Os on each level of a multilevel cache with LRU replacement, under a regularity condition.
This condition says that the I/O complexity $Q(n,M,B)$ satisfies $Q(n,M,B) = \BO{Q(n,2M,B)}$. 
Since our cache-oblivious algorithm for triangle enumeration is optimal and satisfies the regularity condition, we have that this result applies to our algorithm as well.
In the paper, we make the standard tall cache assumption $M=\BOM{B^2}$, which has been shown to be necessary for getting optimal cache-oblivious algorithms, in particular for the problems of sorting~\cite{BrodalF03} and permuting~\cite{Silvestri08}.

We consider a simple, undirected graph (no self loops, no parallel edges) with vertex set $V$ and edge set $E$. 
Each vertex and edge requires one memory word.
For notational convenience and consistency with earlier papers, whenever the context is clear we use $E$ as a shorthand for the {\em size\/} of a set $E$ (and similarly for other sets). 
We denote with $\text{deg}(v)$ the degree of a vertex $v\in V$.
We  assume that the elements of $V$ are ordered according to degree, breaking ties among vertices of the same degree in an arbitrary but consistent way.
We assume that an edge $\{v_1,v_2\}$ is represented by the tuple $(v_1,v_2)$ such that $v_1<v_2$, and that these tuples are sorted lexicographically (so for each vertex $v$ we have the list of neighbors that come after $v$ in the ordering). 
If the graph comes in some other representation, it can be converted to this form in $\sort{E}$ I/Os. 
Following~\cite{HuTC13}, for a triangle  $\{v_1,v_2,v_3\}$, with $v_1<v_2<v_3$, we call the edge $\{v_2,v_3\}$ its {\em pivot edge}, and the vertex $v_1$ its {\em cone vertex}.

The following lemma describes a subroutine that is  widely used in the paper for enumerating all triangles containing  a given vertex $v$. 
\begin{lemma}\label{lem:sortHD}
Enumerating all triangles in an edge set $E$ that contain a given vertex $v$ can be done in $\BO{\text{sort}(E)}$ I/Os.
\end{lemma}
\begin{proof}
By scanning $E$, we find the set $\Gamma_v$ of vertices that are adjacent to $v$, and we sort it by degree. 
Then we sort edges in $E$ by the smallest vertex and find the set $E_v\subseteq E$ of edges with the smallest vertex in $\Gamma_v$, just by scanning $E$ and $\Gamma_v$. 
Finally, we sort edges in $E_v$ by the largest vertex and compute the set of edges $E_v'\subseteq E_v$ with both vertices in $\Gamma_v$ with another scan of $E_v$ and  $\Gamma_v$. 
By construction we have that, for each $e=\{u,w\}\in E_v'$, there exists a triangle with vertices $v,u$ and $w$.
\end{proof}
 
Another subroutine used in the paper is the algorithm given in~\cite{HuTC13} that efficiently finds all triangles with a pivot edge in a set $E'\subseteq E$.
Though this subroutine was presented in~\cite{HuTC13} as a listing algorithm, it is easy to see that it works for enumeration as well. We sketch the result below for the sake of completeness.
\begin{lemma}(Hu et al.~\cite[Algorithm 1, step~2]{HuTC13})\label{lemma:pivot}
The set of triangles in an edge set $E$ with a pivot edge in $E'\subseteq E$ can be enumerated in $\BO{E/B+E'E/(MB)}$ I/Os.
\end{lemma}
\begin{proof}
The algorithm runs in iterations. In each iteration $\alpha M$ new edges from $E'$, for a suitable constant $\alpha\in (0;1)$, are loaded into internal memory. Let $\Gamma_\text{mem}$ be the set of vertices that appear in an edge of $E'$ currently stored in internal memory. Then, for each vertex $v$ in the graph, the algorithm computes the set
$$\Gamma_v=\{u \; \vert \; (v,u)\in E, u>v, u\in \Gamma_\text{mem}\},$$
that is, the set containing all vertices larger than $v$ that are adjacent to $v$, and appear in an edge of $E'$ stored in internal memory in the current iteration. Then, it enumerates all triangles $\{v,u,w\}$ where $\{u,w\}\in E'$ and $u,w\in \Gamma_v$. 
It is easy to see that it is possible to compute $\Gamma_v$ for every vertex $v$ using a single scan of all edges in $E$, since all edges $\{v, u\}\in E$ with $u>v$ are stored consecutively in external memory.
Then, we get the I/O complexity $\BO{\lceil E'/M \rceil \sum_{v\in V} \deg(v)/B}$ which is upper bounded by $\BO{E/B+E'E/(MB)}$.
\end{proof}

\section{Cache-aware enumeration}\label{sec:aware}

Our first algorithm is cache-aware, that is, it is given information on the internal memory size $M$ and on the block length $B$. The algorithm also explicitly manages block transfers. 
Without loss of generality we assume that $E>M$ and that $\sqrt{E/M}$ is an integer.

\subsection{Algorithm overview}

Let $V_h = \{ v \;|\; \deg(v) > \sqrt{E M} \}$ be the set of \emph{high-degree} vertices, and $V_l = V\backslash V_h$ be the remaining \emph{low-degree} vertices. There cannot be too many vertices in the set $V_h$: indeed we have $V_h < \sqrt{E/M}$.
We denote with $E_h$ the set of edges incident to at least one vertex in $V_h$, and with $E_l=E\backslash E_h$ the remaining edges.

The first step of our algorithm enumerates the triangles that involve at least one edge from $E_{h}$ using the algorithm described in Lemma~\ref{lem:sortHD} for each high-degree vertex in $V_h$.
Subsequent steps can then focus on triangles within $E_l$.
Our algorithm will work with a coloring  $\xi: V\rightarrow \{1,\dots,c\}$ of the vertex set where the number of colors will be $c=\sqrt{E/M}$.
The coloring will partition the edges of $E_l$ into $c^2 = E/M$ sets according to the colors of their vertices.
More specifically, for $\tau_1,\tau_2\in \{1,\dots,c\}$ let
$$E_{\tau_1,\tau_2} = \{ \{v_1,v_2\}\in E_l \; | \; v_1 < v_2,\, \xi(v_1)=\tau_1,\, \xi(v_2)=\tau_2\} \enspace .$$

Since the number of partitions is $E/M$, the average number of edges in a partition is $M$.
If all partitions did indeed have size $M$, we could easily obtain an algorithm with the desired I/O complexity by considering all $c^3$ possible coloring of the vertices of a triangle in $\BO{M/B}$ I/Os.
However, some partitions may be much larger than $M$, so there is no guarantee that we can fit a large part of a partition in memory.

\medskip
\noindent
We are now ready to describe the high-level algorithm:
\begin{enumerate}
\item Enumerate all triangles with at least one vertex in $V_h$ using the algorithm of Lemma~\ref{lem:sortHD}.
\item Choose $\xi$ uniformly at random from a 4-wise independent family of functions, and construct the sets $E_{\tau_1,\tau_2}$ using a sorting algorithm.
\item For every triple $(\tau_1,\tau_2,\tau_3) \in \{1,\dots,c\}^3$, enumerate all triangles with a cone vertex of color $\tau_1$ and a pivot edge in $E_{\tau_2,\tau_3}$. We use the algorithm in Lemma~\ref{lemma:pivot} by setting the pivot edge to $E_{\tau_2,\tau_3}$, the edge set  to $E_{\tau_1,\tau_2} \cup E_{\tau_1,\tau_3}\cup E_{\tau_2,\tau_3}$, and ignoring triangles where the cone vertex does not have color $\tau_1$.
\end{enumerate}

\subsection{Analysis}

\subsubsection{Correctness}
We first argue for correctness of the algorithm. 
Every triangle that includes at least one vertex in $V_h$ is enumerated in step 1 by Lemma~\ref{lem:sortHD}. On the other hand, a triangle with vertices $v_1<v_2<v_2$, none of which belongs to $V_h$, is enumerated in step 3, specifically in the iteration where $(\tau_1,\tau_2,\tau_3)=(\xi(v_1),\xi(v_2),\xi(v_3))$.

\subsubsection{I/O complexity}
We define the random variable $X_\xi$ as follows:
\begin{align}\label{colldef}
X_\xi = \sum_{\tau_1,\tau_2} \binom{E_{\tau_1,\tau_2}}{2} \enspace .
\end{align}
This variable denotes the number of pairs of edges in each partition and  will be used for bounding the I/Os in step 3. We have the following bound.

\begin{lemma}\label{lemma:expected}
Let $\xi: V\rightarrow \{1,\dots,c\}$ be chosen uniformly at random from a 4-wise independent family of functions, where $c=\sqrt{E/M}$. Then 
$$\E{X_\xi} \leq \binom{E}{2}/c^2 + \sum_{v\in V_l} \binom{\deg(v)}{2}/c \leq EM.$$
\end{lemma}
\begin{proof}
	Define the indicator variable $Y_{e_1,e_2}$ to be 1 if $e_1$ and $e_2$  are colored in the same way (i.e., belong to the same set $E_{\tau_1,\tau_2}$), and zero otherwise. 
	By linearity of expectation we have:
$$
\E{X_\xi} = \sum_{e_1\ne e_2} \E{Y_{e_1,e_2}} = \sum_{e_1\ne e_2} \Pr{Y_{e_1,e_2}=1}.
$$
There are at most $\sum_{v\in V_l}\binom{\deg(v)}{2}$ pairs of edges $\{e_1,e_2\}\subseteq E_l$ that share a vertex, and for those $\E{Y_{e_1,e_2}} \leq 1/c$ since the $\xi$ function is chosen uniformly at random from a 4-wise independent family of functions.
For the remaining at most $\binom{E}{2}$ pairs the colorings are independent, and hence the probability of having the same coloring is $1/c^2$.
Summing up, and using the fact that $\deg(v) \leq \sqrt{E M}$ for all $v\in V_l$ gives $X_\xi < E^2/(2c^2) + E\sqrt{E M}/(2c)$.
Finally, inserting $c=\sqrt{E/M}$ yields the stated bound.
\end{proof}

\begin{restatable}{thm}{thcacheaware}
\label{th:caware}
Assume $E\geq M\geq {E^{\epsilon}}$, for an arbitrary constant $\epsilon>0$. Then the above cache-aware randomized algorithm for triangle enumeration requires $\BO{\frac{E^{3/2}}{\sqrt{M}B}}$  I/Os in expectation and $\BO{E}$ words on disk.
\end{restatable}
\begin{proof}
When $M\geq {E^{\epsilon}}$, for an arbitrary constant $\epsilon>0$,  the first and second steps together require $\BO{V_h\sort{E}}=\BO{\frac{E^{3/2}}{\sqrt{M}B}}$ I/Os, which is upper bounded by the claimed complexity. By setting $E_{\tau_1,\tau_2,\tau_3}=E_{\tau_1,\tau_2} + E_{\tau_1,\tau_3} + E_{\tau_2,\tau_3}$, we get that the I/O complexity $Q(E,M,B)$ of step 3 is by Lemma~\ref{lemma:pivot}
$$
Q(E,M,B)=\BO{\sum_{(\tau_1,\tau_2,\tau_3)} \frac{E_{\tau_1,\tau_2,\tau_3}}{B} + \frac{E_{\tau_1,\tau_2,\tau_3}^2}{MB}}.
$$
Since $\sum_{(\tau_1,\tau_2)}E_{\tau_1,\tau_2}=E$, the above bound becomes 
$$
Q(E,M,B)=\BO{ \frac{cE}{B} + \hspace{-0.5em} \sum_{(\tau_1,\tau_2,\tau_3)} \hspace{-0.5em}\frac{E^2_{\tau_1,\tau_2} + E^2_{\tau_1,\tau_3} +  E^2_{\tau_2,\tau_3}}{MB}},
$$
and hence
$$
Q(E,M,B)= \BO{ \frac{cE}{B} +\frac{c}{MB} \sum_{(\tau_1,\tau_2)} E^2_{\tau_1,\tau_2}}.
$$
Since $E^2_{\tau_1,\tau_2}\leq 4\binom{E_{\tau_1,\tau_2}}{2}$ and by the definition of $X_\xi$ in~(\ref{colldef}), we get
$$
Q(E,M,B)= \BO{\frac{E^{3/2}}{\sqrt{M}B} +\frac{\sqrt{E}}{M^{3/2}}  X_\xi}.
$$

That is, the expected time complexity is governed by the expectation of the random variable $X_\xi$. By Lemma~\ref{lemma:expected}, we have that $\E{X_\xi}\leq EM$ and then the expected I/O complexity of step 3 is $\BO{E^{3/2} / (\sqrt{M} B)}$.  
The algorithm clearly requires $\BO{E}$ space on disk.
\end{proof}


\renewcommand{\a}{\overline{c_0}}
\renewcommand{\b}{\overlinwe{c_1}}
\renewcommand{\c}{\overl(vine{c_2}}

\section{Cache-oblivious  enumeration}\label{sec:cobl}

In this section we describe a cache-oblivious, randomized algorithm for the enumeration of all triangles in a graph 
in  $\BO{E^{3/2}/(\sqrt{M}B)}$ expected I/Os,
proving Theorem~\ref{th:cobl}. 
Optimality of this bound is shown in Section~\ref{sec:lb}.
As already noticed, an optimal cache-oblivious algorithm implies an optimal number of I/Os on each level of a multilevel cache with LRU replacement if a regularity condition is verified (i.e.,  the I/O complexity of the algorithm $Q(n,M,B)$ satisfies $Q(n,M,B) = \BO{Q(n,2M,B)}$). 
Since our cache-oblivious algorithm is optimal and satisfies the regularity condition, we have that this result applies to our algorithm as well.

\subsection{Algorithm overview}

The cache-oblivious algorithm in this section is inspired by a recursive approach proposed by Jafargholi and Viola~\cite{JafargholiV13}, in the context of output sensitive triangle listing in a RAM model.
To describe the algorithm we define the more general $(c_0,c_1,c_2)$-enumeration problem. 
Let $\xi: V\rightarrow \mathbb{Z}$ be a coloring of the vertex set, assigning an integer to each vertex.
The \emph{$(c_0,c_1,c_2)$-enumeration problem} with coloring $\xi$ consists of enumerating all triangles colored according to the vector $(c_0,c_1,c_2)$, i.e.,~triangles with vertices $\{u,v,w\}\subseteq V$ where $u < v < w$, $\xi(u)=c_0$, $\xi(v)=c_1$, and $\xi(w)=c_2$. The enumeration of all triangles simply reduces to the $(1,1,1)$-enumeration problem with the constant coloring $\xi(v)=1$.

A triangle is \emph{proper} if it satisfies the $(c_0,c_1,c_2)$ coloring, and an edge $\{u,v\}$, with $u< v$, is  \emph{incompatible} with  coloring $(c_0,c_1,c_2)$ if $(\xi(u), \xi(v))\not\in\{(c_0,c_1),(c_1,c_2),(c_0,c_2)\}$. 
Without loss of generality, we assume that there are no incompatible edges in $G$ and that the color of each vertex is stored within the vertex (these assumptions can be guaranteed by suitably sorting edges without increasing the I/O complexity).

Our algorithm solves the $(c_0,c_1,c_2)$-enumeration problem with coloring $\xi$ in three steps: 
\begin{enumerate}
\item   The algorithm enumerates all triangles satisfying the $(c_0,c_1,c_2)$ coloring with at least one local high degree vertex.  A \emph{local high degree vertex} is a vertex with degree at least $E/8$; there are  at most $16$ local high degree nodes. For each local high degree vertex $v$, the algorithm enumerates all triangles containing $v$ with the subroutine in Lemma~\ref{lem:sortHD} (using any efficient cache-oblivious sorting algorithm, e.g.,~the one from~\cite{FrigoLPR12}). Local high degree vertices and their edges are then removed.

\item  A new coloring $\xi': V \rightarrow \mathbb{Z}$ is defined by adding a random bit to the value returned by $\xi$ in the least significant position of the binary representation. Specifically, let $\xi'(v) = 2 \xi(v) - b(v)$, where $b: V \rightarrow \{0,1\}$ is chosen uniformly at random from a 4-wise independent family of functions.

\item  The remaining triangles that satisfy $(c_0,c_1,c_2)$  under  coloring $\xi$ are enumerated by recursively solving 8 subproblems. 
For each color vector $\zeta\in \{2c_0 - 1, 2c_0\} \times \{2c_1 - 1, 2c_1\} \times \{2c_2 - 1, 2c_2\}$,
we recursively solve the $\zeta$-enumeration problem with coloring $\xi'$ on the graph obtained by removing edges incompatible with the color vector.
\end{enumerate}

The recursion ends when $E$ is empty, or at depth $\log_{4} E$: 
in the first base case there are no triangles; in the second base case, triangles are enumerated with the deterministic algorithm by~ Dementiev~\cite{Dementiev2007}, which relies on sort and scan operations, and can be trivially made oblivious using any oblivious sorting algorithm.
We note that step 1 has an effect also in the recursive calls, since $E$ refers to the number of edges compatible with the given subproblem.
In fact, this is the main conceptual difference between our algorithm and the algorithm in~\cite{JafargholiV13}.

We observe that at the recursive level $i=\log c$, with $c =\sqrt{E/M}$, the behavior of the algorithm is similar to the one of the cache-aware algorithm presented in Section~\ref{sec:aware}:
There are $c$ colors
and, as we will see below, when $i = \log c$ each vertex with degree at least $\sqrt{EM}$ is expected to be removed, and there are $(E/M)^{3/2}$ subproblems, each of expected size~$M$.

\subsection{Analysis}

\subsubsection{Correctness}
We argue that all proper triangles with coloring $(c_0,c_1,c_2)$ are correctly enumerated. 
Indeed, proper triangles with a local high degree vertex $v$ are found in step 1, and cannot appear again since edges adjacent to $v$ are subsequently removed. 
The remaining triangles are enumerated in the subproblems. 
Indeed, each proper triangle is given a coloring in $\{2c_0 - 1, 2c_0\} \times \{2c_1 - 1, 2c_1\} \times \{2c_2 - 1, 2c_2\}$ under $\xi'$, and there is exactly one recursive call reporting each triangle.

\subsubsection{I/O Complexity}
Suppose the $8^i$ subproblems at level $i$, with $0\leq i \leq \log_{4} E$, are arbitrarily numbered. We denote by $E_{i,j}$, $\xi_{i,j}$, $(c^0_{i,j}, c^1_{i,j}, c^2_{i,j})$ the input edge set, the coloring, and the triplet defining proper triangles, respectively, of the $j$th subproblem at 
level $i$,  for any $0\leq i \leq \log_{4} E$ and $0\leq j <8^i$.  We then define $E_{i,j}^{k,l}$, for any $0\leq k <l\leq 2$, as the set containing each edge $\{u,v\} \in E_{i,j}$, with $u<v$, such that 
$\xi_{i,j}(u)=c^k_{i,j}$ and $\xi_{i,j}(v)=c^l_{i,j}$. With a slight abuse of notation, we let $E_{i,j}$ and $E_{i,j}^{k,l}$ also denote the size of the respective sets. Since there are no incompatible edges, we have $E_{i,j}\leq E_{i,j}^{0,1} + E_{i,j}^{1,2} + E_{i,j}^{0,2}$.

In order to upper bound the expected I/O complexity of our algorithm we introduce two lemmas. 
Lemma~\ref{lem:exp} gives an upper bound on the expected value and variance of each subproblem at a given recursive level. 
Then, Lemma~\ref{lem:numsub} uses these bounds to limit the probability that a subproblem is larger than the expected size.

\begin{lemma}\label{lem:exp}
For any $0\leq i\leq \log_{4} E$ and $0\leq j < 8^i$, we have
$$\E{E_{i,j}^{0,1}}\leq \frac{E}{4^i}, \qquad \V{E_{i,j}^{0,1}}\leq \frac{3E^2}{16^i}.$$
The same bounds apply to $E_{i,j}^{1,2}$ and $E_{i,j}^{0,2}$.
\end{lemma}
\begin{proof}
For the sake of the analysis we do not remove local high degree vertices in step 1, but replace them with vertices of degree one. 
Specifically, for any removed vertex $v$ with (local) degree $\text{deg}(v)$, we replace it with $\text{deg}(v)$ new vertices $v_i$ of degree $1$, and replace each edge $\{v,u\}$ with $\{v_i,u\}$ for a suitable $i$. 
This assumption simplifies the analysis since no edges are removed in a recursive level. 
However, correctness is not affected since the new vertices will not be involved in any proper triangle enumerated in recursive calls as they have degree one. 
By symmetry we may focus on $E_{i,j}^{0,1}$, the proofs for $E_{i,j}^{1,2}$ and $E_{i,j}^{0,2}$ being analogous.

We now prove by induction that, at any recursive level $0\leq i\leq \log_4 E$, we have
$\E{E_{i,j}^{0,1}}= {X}/{4^i}$ and $\V{E_{i,j}^{0,1}}\leq {X^2}/{16^i} + 2{X}/{4^i}$, where $X=E_{0,0}^{0,1}=E$. The lemma then follows since $X=E$ and ${X^2}/{16^i} + 2{X}/{4^i}\leq 3{X^2}/{16^i}$ as soon as $i\leq \log_4 X$.
The claim is trivially verified when $i=0$ since we get $\E{E_{0,0}^{0,1}}=E_{0,0}^{0,1}=X$, and $\V{E_{0,0}^{0,1}}=0$.  

Now consider a subproblem $j$ at level $i>0$ and its parent problem $j'$ at level $i-1$. 
By the inductive hypothesis, we have for the parent problem that $\E{E_{i-1,j'}^{0,1}} = X/4^{i-1}$ and $\V{E_{i-1,j'}^{0,1}}\leq  X^2/16^{i-1} + 2 X/4^{i-1}$. 
Assign to each edge $e\in E_{i-1,j'}^{0,1}$ a random variable $Y_e$ equal to one if $e\in E_{i,j}^{0,1}$ and $0$ otherwise. 
By conditioning on the number of edges in the parent problem, we get
\begin{align*}
\E{E_{i,j}^{0,1}} = \E{\EC{\sum_{e\in E_{i-1,j'}^{0,1}} Y_e}{E_{i-1,j'}^{0,1}}} =\frac{\E{E_{i-1,j'}^{0,1}}}{4} =  \frac{X}{4^i}
\end{align*}
since each edge in $E_{i-1,j'}^{0,1}$ is in $E_{i,j}^{0,1}$ with probability $1/4$. The first claim follows. 

Now consider the variance. We have
\begin{align}
\V{E_{i,j}^{0,1}} &= \E{(E_{i,j}^{0,1})^2}- \E{E_{i,j}^{0,1}}^2 = \E{\EC{(E_{i,j}^{0,1})^2}{E_{i-1,j'}^{0,1}}}- \E{E_{i,j}^{0,1}}^2.\label{eq:var}
\end{align}
The conditional expectation $\EC{(E_{i,j}^{0,1})^2}{E_{i-1,j'}^{0,1}}$ can be computed as follows. Since we have $E_{i,j}^{0,1} = \sum_{e\in E_{i-1,j'}^{0,1}} Y_e$, it follows that
$$
\left(E_{i,j}^{0,1}\right)^2=
\sum_{e\in E_{i-1,j'}^{0,1}} Y_e^2  + \sum_{\substack{e, e'\in E_{i-1,j'}^{0,1}\\ e\neq e', \vert e\cap e'\vert=1}} Y_e Y_{e'} + \sum_{\substack{e, e'\in E_{i-1,j'}^{0,1}\\ e\neq e', e\cap e'=\emptyset}}  Y_e Y_{e'}.
$$
Let $W_{i,j}^{0,1}=\sum_{e, e'\in E_{i-1,j'}^{0,1}, e\neq e', e\cap e'=\emptyset}  Y_e Y_{e'}$;
note that $W_{i,j}^{0,1}$ denotes the number of edge pairs that do not share any vertex in the $j$-th subproblem at level $i$.
Two edges sharing a vertex (i.e., $\vert e \cap e'\vert =1$) are in $E_{i,j}^{0,1}$ with probability $1/8$. Also, there are at most $2(E_{i-1,j'}^{0,1})^2/8$  pairs of edges that share exactly one vertex, since the maximum degree after step 2 is $E_{i-1,j'}^{0,1}/8$. Then the conditional expectation becomes
\begin{align}
&\EC{(E_{i,j}^{0,1})^2}{E_{i-1,j'}^{0,1}}=\frac{E_{i-1,j'}^{0,1}}{4} + \frac{(E_{i-1,j'}^{0,1})^2}{32} + \EC{W_{i,j}^{0,1}}{E_{i-1,j'}^{0,1}}\label{eq:condeq},
\end{align}
We then take  the expectation of~(\ref{eq:condeq}):
\begin{align*}
&\E{\EC{(E_{i,j}^{0,1})^2}{E_{i-1,j'}^{0,1}}}= \\
&=\frac{\E{E_{i-1,j'}^{0,1}}}{4} + \frac{\E{(E_{i-1,j'}^{0,1})^2}}{32} + \E{W_{i,j}^{0,1}}\\
&=\frac{\E{E_{i-1,j'}^{0,1}}}{4} + \frac{\V{E_{i-1,j'}^{0,1}}+ \E{(E_{i-1,j'}^{0,1})}^2}{32}+ \E{W_{i,j}^{0,1}}.
\end{align*}
By the inductive hypotheses on expectation and variance at level $i-1$ it follows that
\begin{align*}
\E{\EC{(E_{i,j}^{0,1})^2}{E_{i-1,j'}^{0,1}}}\leq \frac{X^2}{16^{i}}+\frac{5X}{4^{i+1}} +  \E{W_{i,j}^{0,1}}.
\end{align*}
The term $\E{W_{i,j}^{0,1}}$ can be upper bounded assuming that no two input edges share a vertex in $E_{0,0}^{0,1}$. This  gives an upper bound since a vertex shared by two edges cannot increase $W_{i,j}^{0,1}$. 
By induction it follows that $\E{W_{i,j}^{0,1}}\leq X^2/16^i$: 
indeed, an edge pair in $E_{i-1,j'}^{0,1}$ is also in $E_{i,j}^{0,1}$ with probability $1/16$. Then we get
\begin{align*}
\E{\EC{(E_{i,j}^{0,1})^2}{E_{i-1,j'}^{0,1}}}
\leq \frac{2X^2}{16^i}+\frac{5X}{4^{i+1}}.
\end{align*}
Finally, by~(\ref{eq:var}), we get that the variance at level $i$ is:
\begin{align*}
\V{E_{i,j}^{0,1}} \leq \frac{2X^2}{16^i}+\frac{5X}{4^{i+1}}-\frac{X^2}{16^{i}} \leq  \frac{X^2}{16^i}+\frac{2X}{4^{i}}
\end{align*}
and the claim follows. 
\end{proof}

\begin{lemma}\label{lem:numsub}
For any $0\leq i\leq \log_{4} E$, $0\leq j < 8^i$ and $0\leq k < \log_{4} E-i$, we have that
$$
\Pr{E_{i,j}\geq 9 \frac{E}{4^{i-k}}} \leq  {1}/{16^k}.
$$
\end{lemma}
\begin{proof}
Since $E_{i,j}\leq E_{i,j}^{0,1}+E_{i,j}^{1,2}+E_{i,j}^{0,2}$, we clearly have $\Pr{E_{i,j}\geq \beta \frac{E}{4^{i-k}}}\leq 3 \Pr{E_{i,j}^{0,1}\geq  \frac{\beta}{3}\frac{E}{4^{i-k}}}$. 
 Lemma~\ref{lem:exp} gives $\E{E_{i,j}^{0,1}}\leq E/4^i$ and $\V{E_{i,j}^{0,1}}\leq 3E^2/(16)^i$. Then, by Chebyshev's inequality, we get
\begin{align*}
\Pr{E_{i,j}^{0,1}\geq \frac{\beta}{3} \frac{E}{4^{i-k}}} & \leq 
\Pr{ \left\vert E_{i,j}^{0,1} - \E{E_{i,j}^{0,1}} \right\vert \geq  \frac{({\beta}/{3}-1)E}{4^{i-k}}}\\
&\leq 9 \frac{\V{E_{i,j}^{0,1}}16^{i-k}}{(\beta-3)^2 E^2} \leq \frac{27}{(\beta-3)^2  16^k}.
\end{align*} 
By setting $\beta=9$ the lemma follows.
\end{proof}

We are now ready to prove the first result of the paper, repeated here for convenience.
\thcacheoblivious*
\begin{proof}
We first argue that the I/O complexity of subproblems with input size not larger than $M$ is asymptotically negligible. 
Consider a subproblem $x$ whose input size is smaller than $M$, but its parent $y$ has input size larger than $M$.
Since the data used by $x$ fits in memory, the I/O complexity for solving $x$ (including subproblems generated in $x$) is $\BO{M/B+1}$. On the other hand, in our analysis we assume that the I/O complexity of $y$ is $\BOM{M/B+1}$, and thus the cost for solving $x$ is asymptotically negligible. Since a problem with input larger than $M$ can have at most 8 child  subproblems,  we can ignore subproblems of size smaller than $M$ without affecting asymptotically the I/O complexity of our algorithm.

We now upper bound the I/O complexity without taking into account the cost of subproblems at level $\log_4 E$ which have a slightly different I/O complexity than a subproblem at level $i<\log_4 E$ --- we will later see how to bound this quantity.

Let $Y_{i,s}$ denote the number of subproblems at level $i$ with input size $(E/4^{s+1}, E/4^{s}]$, for any $0\leq i\leq \log_4 E$ and $0\leq s < \log_{4} (E/M)$. 
The cost of a subproblem of size $Y_{i,s}$ is dominated by the sorting in step 1, and  then we get:
\begin{align*}
Q(E,M,B)&=\BO{\sum_{i=0}^{\log_{4}{E}-1} \sum_{s=0}^{ \log_4 (E/M)} Y_{i,s}\text{sort}(E/4^s)} \\
&= \BO{\sum_{s=0}^{ \log_4 (E/M)} \sum_{i=0}^{\log_{4} E-1} Y_{i,s} \text{sort}(E/4^s)}
\end{align*}
Since there are at most $2 \cdot 8^s$ subproblems of size no larger than  $E/4^{s}$ at levels $0, \ldots, s$, we get
\begin{align*}
Q(E,M,B)= \BO{\sum_{s=0}^{ \log_4 (E/M)} \left(8^s+\sum_{i=s+1}^{\log_{4} E-1}  Y_{i,s} \right) \text{sort}(E/4^s)}.
\end{align*}
By Lemma~\ref{lem:numsub},
the probability that a subproblem at level $i>s$ has size at least $E/4^{s+1}$ is 
$$\Pr{E_{i,j}\geq {E}/{4^{s+1}}} \leq \Pr{E_{i,j}\geq {9 E}/{4^{s+3}}} \leq 1/16^{i-s-3} \enspace .$$
The expected number of subproblems  of size larger than $E/4^{s+1}$ at level $i>s$ is 
$8^{i}\frac{1}{16^{i-s-3}}=\BO{16^s/2^i}$, which means that
$$\E{Y_{i,s}}=\BO{16^s/2^i} \text{ and }\E{\sum_{i=s+1}^{\log_{4} E} Y_{i,j}}=\BO{8^s} \enspace .$$ 
It follows that the expected value of $Q(E,M,B)$ is
\begin{align*}
&\E{Q(E,M,B)} = \BO{\sum_{s=0}^{ \log_4 (E/M)} 2^s \frac{E \log (E/4^s)}{B \log M}} \\
&= \BO{\frac{E}{B\log M}\int_0^{\log_4 (E/M)+1} 2^x \log (E/4^x) dx}\\
&=  \BO{\frac{E}{B\log M} \left. \frac{2^x \left(\ln (E/4^x)+2\right) 
}{\ln^2 2}\right\vert _0^{\log_4 (E/M)+1}}\\
&=  \BO{\frac{E^{3/2}}{\sqrt{M}B}}.
\end{align*}
We now bound the  expected number of I/Os required for subproblems at level $i=\log_4 E$.
Since we are using the algorithm by Dementiev~\cite{Dementiev2007} for solving base cases, the cost of a subproblem with input size in the range $(E/4^{s+1}, E/4^{s}]$ is $\BO{\text{sort}\left((E/4^s)^{3/2}\right)}$  I/Os.
This means that the number $Q'(E,M,B)$ of I/Os required by level $\log_4 E$ is:
\begin{align*}
Q'(E,M,B)= \BO{\sum_{s=0}^{\log_4 (E/M)} Y_{\log_4 E, s}\text{sort}\left((E/4^s)^{3/2}\right)}.
\end{align*}
By applying Lemma~\ref{lem:numsub} as before, we get that 
$$\E{Y_{\log_4 E, s}}=\BO{16^s/2^{\log_4 E}} \enspace .$$ 
Hence the expected value of $Q'(E,M,B)$ is
\begin{align*}
\E{Q'(E,M,B)}= \BO{\sum_{s=0}^{\log_4 (E/M)} 
2^{s}\frac{E  {\log (E/4^s)}}{B\log M}}
\end{align*}
which is $\BO{\frac{E^{3/2}}{\sqrt{M}B}}$ as shown before. The I/O complexity of the cache-oblivious algorithm follows by summing the expected values of $Q(E,M,B)$ and $Q'(E,M,B)$.

If the input of a subproblem is stored in a new location, the used space on disk is $\BO{E}$ in expectation since the expected size decreases geometrically. However, $\BO{E \log E}$ space is required in the worst case (i.e., when there exists only one partition containing all edges). The claimed $\BO{E}$ bound follows by noticing that no new space is required for storing subproblem input: before each recursive call, edges are sorted so that the subproblem input is stored in consecutive locations in the input of the parent problem. In this case, just  pointers to the initial and final positions are required for denoting the input. 
\end{proof}


\section{Derandomization} 
\label{sec:derand}
We now pursue a derandomization of the cache-aware algorithm in Section~\ref{sec:aware} via small-bias probability spaces.
More specifically, we need to find a balanced coloring $\xi$ such that $X_\xi = \BO{EM}$.
The idea of using this method to derandomize a triangle enumeration algorithm was previously used in~\cite{JafargholiV13}, though we present a more refined greedy approach that preserves the exponent $3/2$ of the algorithm.

For convenience we round up the number of colors $c$ to the nearest power of~2, which can only decrease $\E{X_\xi}$ for random $\xi$.
We split $X_\xi$ into two terms, $X_\xi = X^{\text{adj}}_\xi + X^{\text{nonadj}}_{\xi}$, where the two terms are the contributions in the sum defined in~(\ref{colldef}) from adjacent and non-adjacent edge pairs, respectively.

Our algorithm fixes one bit of the coloring at a time, aiming to approach the coloring guarantee of Lemma~\ref{lemma:expected}.
Formally we start with the constant coloring $\xi_0$ that assigns color~1 to every vertex.
For $i = 1,\dots,\log c$ we find a two-coloring $b_{i-1}: V\rightarrow \{0,1\}$ such that the coloring $\xi_i(v) = 2\xi_{i-1}(v)-b_{i-1}(v)$ satisfies 
\begin{align}\label{collbound}
\frac{4^i X^{\text{nonadj}}_{\xi_i}}{c^2} + \frac{2^i X^{\text{adj}}_{\xi_i}}{c} \leq (1+\alpha)^i EM \enspace .
\end{align}
Setting $\alpha = 1/\log c$ we have $(1+\alpha)^{\log c} < e$ so for the final coloring $\xi = \xi_{\log c}$, since $c\geq \sqrt{E/M}$ we get
$X_{\xi_i} =  X^{\text{nonadj}}_{\xi_i} + X^{\text{adj}}_{\xi_i} < e EM$.

It remains to be shown how we select $\xi_i$ to ensure~(\ref{collbound}) for $i=0,\dots,\log c$.
For $i=0$ we have $X^{\text{nonadj}}_{\xi_0} < E^2/2$ and $X^{\text{adj}}_{\xi_0} < E \sqrt{EM} /2$, and inserting $c\geq \sqrt{E/M}$ the claim follows.
The function $b_{i-1}$ used for constructing $\xi_i$ for $i>0$ will be taken from an almost 4-wise independent sample space.
We use the following known result:
\begin{lemma} (\cite[Theorem 2]{AGHP92}.) \label{lemma:smallbias}
For any $\alpha>0$ there is a set of $t = \BO{(\log(V) / \alpha)^2}$ functions $\beta_1,\dots,\beta_t: V\rightarrow \{0,1\}$ such that: For every four vertices $v_1,v_2,v_3,v_4$ and each vector $x\in \{0,1\}^4$ the set $\{ \beta_j \; | \; (\beta_j(v_1),\beta_j(v_2),\beta_j(v_3),\beta_j(v_4))=x \}$ has size at most $(1+\alpha) 2^{-4} t$. The space required for computing a value of $b_i$ is $\BO{\log(V/\alpha)}$ bits.
\end{lemma}

We now argue that if~(\ref{collbound}) holds for $\xi_{i-1}$ there exists a function $b_{i-1}$ from the sample space of Lemma~\ref{lemma:smallbias} such that~(\ref{collbound}) holds for $\xi_{i}$. 
To see this, consider the function 
$$\xi_{i}(v) = 2\xi_{i-1}(v) - b_{i-1}(v)$$
where $b_{i-1}$ is chosen at random from the family of Lemma~\ref{lemma:smallbias}.
Then $\E{X^{\text{nonadj}}_{\xi_{i}}} \leq X^{\text{nonadj}}_{\xi_{i-1}} (1+\alpha)/4$ because each pair contributing to $X^{\text{nonadj}}_{\xi_{i-1}}$ has probability at most $(1+\alpha)/4$ of colliding under $\xi_{i}$.
Similarly, $\E{X^{\text{adj}}_{\xi_{i}}} \leq X^{\text{adj}}_{\xi_{i-1}} (1+\alpha)/2$.
This means that
\begin{align*}
&\E{\frac{4^i X^{\text{nonadj}}_{\xi_{i}}}{c^2} + \frac{2^i X^{\text{adj}}_{\xi_{i}}}{c}} 
=
\frac{4^i \E{X^{\text{nonadj}}_{\xi_{i}}}}{c^2} + \frac{2^i \E{X^{\text{adj}}_{\xi_{i}}}}{c}\leq \\
& \hspace{1em}
 \leq (1+\alpha) \left(\frac{4^{i-1} X^{\text{nonadj}}_{\xi_{i-1}}}{c^2} + \frac{2^{i-1} X^{\text{adj}}_{\xi_{i-1}}}{c}\right) \leq  (1+\alpha)^i EM.
\end{align*}
So we conclude that there must exist a choice of $b_{i-1}$ for which~(\ref{collbound}) holds.

Finally, we need to argue that the right function $b_{i-1}$ can be chosen efficiently. 
To do this we maintain the list of edges sorted according to color class, such that all edges in
$$
E^{i-1}_{\tau_1,\tau_2} = \{ \{v_1,v_2\}\in E_l | v_1 < v_2,\, \xi_{i-1}(v_1)=\tau_1,\, \xi_{i-1}(v_2)=\tau_2\}$$
are stored consecutively. Since $X_{\xi_{i}} =  \sum_{\tau_1,\tau_2} \binom{E^{i}_{\tau_1,\tau_2}}{2}$, in a single scan of the edge list we can compute the value of~(\ref{collbound}) for every choice of $b_{i-1}$, using the assumption that $M$ is large enough to hold a constant number of variables for each function in internal memory. In particular, what is needed is keeping track of the number of edges of each color class $E^{i-1}_{\tau_1,\tau_2}$ that go into each of the four possible new color classes for those edges. We then select the function $b_{i-1}$ that minimizes~(\ref{collbound}), and split the edge set into new color classes in one additional scan. This concludes the description of our deterministic cache-aware algorithm for triangle enumeration:

\thdeterministic*
\begin{proof}
If $M > k \log^2 V  \log^2(E/M)$ for a sufficiently large constant $k$, we spend $\BO{E/B}$ I/Os for finding the best $b_i$, and then $\BO{\text{sort}(E)}$ I/Os to organized edges after fixing the coloring $\xi_i$, for each $i=1,\ldots \log c$. Thus the final balance coloring is computed in $\BO{E \log (E/M) /B}$ I/Os as soon as $M\geq E^\epsilon$. By mimicking the argument of Theorem~\ref{th:caware}, we get the claim since $X_\xi\leq eEM$.
\end{proof}


\section{Lower bound}\label{sec:lb}
In this section, we lower bound the I/O complexity of any algorithm for triangle enumeration. We restrict our attention to algorithms where each edge requires at least one memory word. That is, at any point in time there can be at most $M$ edges in memory, and an I/O can move at most $B$ edges to or from memory. This assumption is similar to the indivisibility assumption~\cite{ArgeM99} which is usually required for deriving lower bounds on the I/O complexity. The optimality of our algorithms follows from the following theorem since a clique of $\sqrt{E}$ vertices has $t=\BOM{E^{3/2}}$ triangles.

\thlowerbound*
\begin{proof}
In order to emit a triangle, information on the three nodes (or edges)  must reside in internal memory at some point in time. 
Since there are at most $M$ edges in internal memory, it follows from e.g.~\cite[Section 4.1]{AfratiSSU13} that no more than $\BO{M^{3/2}}$ distinct triangles can be emitted without doing any I/O.

Let $\mathcal{A}$ be any (possibly non-deterministic) algorithm for triangle enumeration. 
For the sake of the lower bound, consider the best execution $\mathcal{A}'$  of  algorithm $\mathcal{A}$ for a given input graph on a internal memory of size $M$ and block $B$. In other words, we consider the execution getting the smallest I/O complexity  $Q_{\mathcal{A}'}(E,M,B)$ for a given input: for instance, for a randomized algorithm we take the execution with the most favorable  choice of the random values.
Since, $\mathcal{A}'$ is an execution, all decisions that can be made by  algorithm $\mathcal{A}$ have already been taken.

We simulate the execution $\mathcal{A}'$ on an internal memory of size $2M$ in such a way that the computation advances in epochs, and blocks are read (resp., written) on disk only at the beginning (resp., end) of an epoch.  
The simulation works as follows.
We consider the internal memory to be divided into two non-overlapping parts $\mathcal M_0$ and $\mathcal M_1$ of size $M$, where $\mathcal M_0$ will be used to simulate the memory of size $M$ used by execution $\mathcal{A}'$, and $\mathcal M_1$ will be used for anticipating/delaying block reads/writes.
Specifically, each epoch simulates $M/B$ consecutive I/Os of $\mathcal{A}'$: the input blocks are prefetched and stored in $\mathcal M_1$ at the beginning of the epoch; the output blocks are temporary stored in $\mathcal M_1$ and then written on the external memory at the end of the epoch; the I/Os performed by  $\mathcal{A}'$  are then simulated by moving data between $\mathcal M_0$ and $\mathcal M_1$. By construction, we have that the I/O complexity of the simulation is $Q_{\mathcal{A}'}(E,M,B)$ and the I/O complexity of each epoch is $M/B$ (except the last epoch, which may use fewer I/Os).
In an epoch the processor touches at most $2M$ internal memory words and thus $\BO{M^{3/2}}$ distinct triangles can be emitted. Then we have
$$
Q_{\mathcal{A}'}(E,M,B) \geq \left \lfloor \frac{t}{\BO{M^{3/2}}}\right\rfloor \frac{M}{B}.
$$
Since $\BOM{t^{2/3}}$  edges are required for enumerating $t$ distinct triangles, we also have $Q_{\mathcal{A}'}(E,M,B) = \BOM{t^{2/3}/B}$ and the theorem follows.
\end{proof}

\section{Conclusion}\label{sec:concl}
In this paper we have investigated the I/O complexity of triangle enumeration in  external memory. In particular, we have described an optimal cache-oblivious algorithm requiring  $\BO{E^{3/2}/(\sqrt{M} B)}$ expected I/Os, which improves previous bounds by a factor $\min(\sqrt{E/M},\sqrt{M})$. 

Recently, it has been shown~\cite{Silvestri14} that the cache-aware randomized algorithm described in Section~\ref{sec:aware} can be extended to the enumeration of a given subgraph 
with $k$ vertices in the Alon class~\cite{AfratiSSU13} (which includes $k$-cliques) with $\BO{E^{k/2}/(M^{k/2-1}B)}$ expected I/Os if $k\geq 3$ is a constant. 
The algorithm decomposes the problem into $\BO{(E/M)^{k/2}}$ subproblems of expected size $\BO{M}$ using the random coloring  technique in Section~\ref{sec:aware}; each subproblem is then solved using an  extension of the algorithm in~\cite{HuTC13} that  enumerates all cliques of $k$ vertices in $\BO{E^{k-1}/(M^{k-2}B)}$ I/Os. 

An interesting open problem is to derive a triangle enumeration algorithm whose I/O complexity is sensitive to the number of triangles in the input graph.
Another direction is to extend to more general types of database queries, and consider for example cases of cyclic joins where the sizes of relations differ.
Recently Pagh and St{\"o}ckel~\cite{2014arXiv1403.3551P} made progress on I/O-efficient join algorithms that make duplicate-eliminating projections.
Extending their approach to other types of database queries is also an interesting direction.

\subsubsection*{Acknowledgments} The authors would like to thank Konstantin Kutzkov and Thomas Dueholm Hansen for discussions in the early stages of this work, Yufei Tao for providing us with a copy of the extended version of~\cite{HuTC13}, and the anonymous reviewers for useful comments.
This work was supported by the Danish National Research Foundation under the Sapere Aude program,  by MIUR of Italy under project AMANDA, and by the University of Padova under project CPDA121378.


\bibliographystyle{plain}

\bibliography{biblio}

\end{document}